\documentclass[12pt]{article}

\usepackage{url}
\usepackage{bbm}
\usepackage{mathtools}
\usepackage{amssymb}
\usepackage{amsthm}
\usepackage{empheq}
\usepackage{latexsym}
\usepackage{enumitem}
\usepackage{eurosym}
\usepackage{dsfont}
\usepackage{appendix}
\usepackage{color} 
\usepackage[unicode]{hyperref}
\usepackage{frcursive}
\usepackage[utf8]{inputenc}
\usepackage[T1]{fontenc}
\usepackage{geometry}
\usepackage{multirow}
\usepackage{lmodern}
\usepackage{anyfontsize}
\usepackage{pgfplots}
\usepackage{stmaryrd}
\usepackage{float}
\usepackage{bookmark}

\graphicspath{ {images/} }

\pgfplotsset{compat=newest}

\definecolor{red}{rgb}{0.7,0.15,0.15}
\definecolor{green}{rgb}{0,0.5,0}
\definecolor{blue}{rgb}{0,0,0.7}
\hypersetup{colorlinks, linkcolor={red},citecolor={green}, urlcolor={blue}}
      
\makeatletter \@addtoreset{equation}{section}

\newtheorem{problem}{Problem}
\newtheorem{theorem}{Theorem}
\newtheorem{theorem2}{Theorem}[section]

\newtheorem{example}[theorem2]{Example}

\newtheorem{lemma}[theorem2]{Lemma}
\newtheorem{proposition}[theorem2]{Proposition}

\newtheorem{definition}[theorem2]{Definition}
\newtheorem{remark}[theorem2]{Remark}

\DeclareUnicodeCharacter{014D}{\=o}
\usepackage{setspace}
\allowdisplaybreaks

\usepackage{natbib}
\setcitestyle{authoryear,open={(},close={)}} 

\title{Optimal trading without optimal control\footnote{Bastien Baldacci gratefully acknowledge the financial support of the ERC Grant 679836 Staqamof and would like to thank Iuliia Manziuk (Ecole Polytechnique) for fruitful discussions.}}
\author{Bastien {\sc Baldacci}\footnote{\'Ecole Polytechnique, CMAP, 91128, Palaiseau, France, bastien.baldacci@polytechnique.edu.} \and Jerome {\sc Benveniste}\footnote{Courant Institute of Mathematical Science, New York University, 251 Mercer St., New York, NY 10012 } \and Gordon {\sc Ritter}\footnote{Courant Institute of Mathematical Science, New York University, 251 Mercer St., New York, NY 10012 (corresponding author), ritter@post.harvard.edu } }

\begin{document}

\maketitle

\begin{abstract}
A hypothetical risk-neutral agent who trades to maximize the expected profit of the next trade will approximately exhibit long-term optimal behavior as long as this agent uses the vector $p = \nabla V(t,x)$ as effective microstructure alphas, where $V$ is the Bellman value function for a smooth relaxation of the problem. Effective microstructure alphas are the steepest-ascent direction of $V$, equal to the generalized momenta in a dual Hamiltonian formulation. This simple heuristics has wide-ranging practical implications; indeed, most utility-maximization problems that require implementation via discrete limit-order-book markets can be treated by our method. 
\end{abstract}

\section{Introduction}

Consider an investor whose preferences  are described by a utility function of wealth, $u(w)$, as per \cite{arrow1963} and \cite{pratt1964risk}.
Let $w_T$ denote the investor's wealth at some known final time $T$. 
The investor attempts to maximize the expectation of utility of final wealth, $\mathbb{E}[u(w_T)]$, by trading financial assets. 
The mechanism by which buyers meet sellers and trades occur is known as the market microstructure. In this work, the microstructure is assumed to be a continuous double auction electronic order book with time priority, although our methods could be generalized to include other kinds of market microstructure. In continuous limit-order-book microstructure, 
trades are effected by submitting limit orders to an exchange's matching engine. For each security being traded, the investor must determine the price levels at which to submit buy and sell orders and the associated share quantities attached to those orders. In real markets, the price levels are discrete; the minimum possible price increment is the quote resolution allowed by the exchange, known as the tick size. Alternatively, the investor may decide to refrain from placing any orders or cancel some existing orders. Other decision variables include order type and venue. Considering all of these details, we see that the instantaneous action space is an inconveniently large discrete space; we discuss ways of simplifying it later on. 
\begin{problem}\label{full-problem}
The investor seeks the optimal dynamic strategy for choosing an action \mbox{$a_t \in \mathcal{A}_t$} 
at each time $t$, where $\mathcal{A}_t$ is the set of possible actions at time $t$, optimal in the sense of maximizing the expected 
utility of final wealth, $\mathbb{E}[u(w_T)]$.
\end{problem}
Problem \ref{full-problem} 
is mathematically deep and perhaps intractable; it is essentially a 
stochastic optimal control problem over high-dimensional discrete 
action and state spaces. According to \cite{cont2017optimal},  
``Although simultaneous optimization of order timing, type, and routing decisions is an 
interesting problem, it also appears to be intractable\ldots'' and even this is a special case of \mbox{Problem \ref{full-problem}}. \medskip

The purpose of the present paper is to give 
practically implementable methods which execution desks could start using right away to solve Problem \ref{full-problem} approximately. Our approximation method breaks the problem up into two parts. The first part is to construct a \emph{smooth relaxation} of the problem, which
is essentially the continuous-time and continuous-space limit; the second part is to adjust our microstructure decisions to track the smooth relaxation optimally.  The key feature of all smooth relaxations is 
that they hide microstructure details behind smooth cost functions meant to 
represent the average cost of trading at a given rate; they provide 
no guidance on microstructure-level decisions, effectively assuming all executions use market orders. In particular, if we can predict the probability of a passive fill at any given instant (e.g., based on order book imbalance), it is not clear how to use this information in the context of a smooth relaxation, whereas our model provides a very clear and obvious way for the implementor to take advantage of predictions of passive fill completion. \medskip

Let $V(t,x)$ denote the Bellman value function for the smooth relaxation, defined to be the remaining expected-utility gain from time 
$t$ obtained from following the best policy when the current state at time $t$ is $x$. There is also, in principle, a value function for Problem \ref{full-problem} defined in the same way, but the latter
appears to be intractable. The present paper's key idea is to exploit the value function of the smooth relaxation to provide effective microstructure alphas that adjust the microstructure decisions toward long-term optimality. The vector 
\begin{equation} \label{eq:p-nabla-1}
	p := \nabla V(t,x) \in \mathbb{R}^d 
\end{equation}
plays a central role in our approach, where $d\geq 1$ is the number of traded assets. As defined, $p$ is the
direction of steepest ascent for the value function. In our heuristics,  
\eqref{eq:p-nabla-1} encodes all of the information about the long-term 
utility function is needed to make the microstructure decision, so it provides the key link between the trading schedule and the order routing problem. \medskip

In order to describe our policy for selecting the best 
microstructure action, we must first introduce some more notation. 
Write $R_{i,t}(v, a)$ for the (random) profit (or loss, if negative) 
from an order of quantity $v$ on stock $i$ using 
action $a$ over a short interval $[t, t+\delta t]$. 
The action $a$ includes the trader's choice of whether to trade
passively or aggressively. A trader's expected profit 
$\mathbb{E}[R_{i,t}(v, a)]$ depends on 
the trading cost associated with the pair $(v, a)$, and also on 
the trader's views concerning the short-horizon midpoint price return 
\[ 
	r_i^{\text{mid}} = \text{mid}_i(t + \delta t) / \text{mid}_i(t) - 1. 
\]
With no subscript, $r^{\text{mid}}$ denotes the $d$-dimensional vector 
of all midpoint returns for all assets. 

\begin{definition}\label{def:microstructure-alpha}
The term \textbf{effective microstructure alphas}, as used in this paper, will denote 
a set of parameters given to a microstructure trader in the place of 
$\mathbb{E}[r^{\text{mid}}]$, for the purpose of satisfying either 
a short-term goal or a long-term goal. 
\end{definition}

Microstructure alphas, as defined, could be a simple prediction of $r^{\text{mid}}$, or, more interestingly, they could be purposely skewed to encourage trading to increase expected utility (ie. increase long-term alpha and reduce risk) as we shall suggest in Equation \eqref{eq:Er-equals-p}. Our heuristics is, at time $t$, for each security indexed by 
$i\in \{1, \dots, d\}$, choose an instantaneous action $a_i^\star$
which solves the following maximization (over the finite set $\mathcal{A}_{t,i}$ of possible actions on the asset $i$ at time $t$):  
\begin{align} 
	{a_i}^\star &= \mathop{\text{argmax}}_{a \in \mathcal{A}_{t,i}} 
	\, \mathbb{E}[R_{i,t}(v_a, a)]  
	\ \ \text{ for all } \ i \in \{ 1, \dots, d \},
	\label{eq:heuristics}  
	\\ 
	\text{where: } & \mathbb{E}[r^{\text{mid}}] = p \, :=\,  \nabla V(t,x),
	\label{eq:Er-equals-p} 
\end{align}
and follow this action over the interval $[t, t + \delta t)$. 
Here $v_a$ denotes the quantity associated with action $a$; 
for example ``$a = $ aggressive buy 100 shares'' means $v_a = 100$.  
Note also that $r^{\text{mid}}$ and $p$ are $n$-vectors, so the equation 
$\mathbb{E}[r^{\text{mid}}] = p$ expresses the trader's views in all $n$ assets. \medskip

There is a very good intuitive justification for \eqref{eq:heuristics}-\eqref{eq:Er-equals-p}.  We show later that, under certain conditions, 
the optimal instantaneous trading rate at any time $t$ in the smooth relaxation is given by  
\begin{equation}\label{intuition-argmax-v}
	\mathop{\text{argmax}}_v \{ \langle p, v \rangle - c(v) \},
\end{equation}
where $c(v)$ is the average cost of trading at rate $v$. 
The expression \mbox{$\langle p, v \rangle - c(v)$} is the instantaneous analogue of expected profit minus cost, if your expected return is $p$. \medskip

The rest of the paper is organized as follows. In Section \ref{sec_AC_curve}, we present an example of a long-term trading schedule, the Almgren-Chriss case, and how to compute the value function $V(t,x)$ and its gradient $p$.  We further show how $p$ is related to the generalized momenta of the Hamiltonian approach. In Section \ref{sec_smooth_relaxation}, we derive the heuristics \eqref{eq:heuristics}  by analogy with the smooth case and show how the trader can choose its short-term alpha in order to minimize the error with respect to the trading schedule. In Section \ref{sec_general_microstructure}, we present a general microstructure trading framework on a portfolio of cross-listed assets, taking into account long and short-term trading signals as well as many components of market microstructure (spread, imbalance, probability of filling etc.). We also show how this heuristics can be applied to the problem of multi-asset market-making. Finally, Section \ref{sec_numerical_results} shows a detailed numerical example  which illustrates the dangers of separating portfolio construction from execution, and in which our method generates an improvement that is both 
statistically and economically significant.

\section{The long-term trading curve}\label{sec_AC_curve}

In this section, we recall the optimal portfolio liquidation framework of \cite{almgren2001optimal} in continuous time and show how to solve the problem using the Hamiltonian method. In the process, we introduce notation that is used in the rest of the paper. \medskip

Variations of the Almgren-Chriss model are used in execution desks all over the world. Indeed it is safe to say that the Almgren-Chriss model (usually with some custom extensions) has been used to execute trillions of dollars' worth of customer orders. We also believe that the Almgren-Chriss model, used in conjunction with the presented heuristics for microstructure decisions, is a reasonable choice if the trader has no alpha forecasts (with the possible exception of microstructure forecasts). Moreover, the computations we do here generalize in a straightforward way to extensions of the Almgren-Chriss model, such as to include long-term alpha. One of the key pieces of intuition that allows us to generalize the model, we feel, comes from Theorem \ref{thm:value-function-gradient} which makes the connection between microstructure alpha, Hamilton's generalized momentum, and the steepest-ascent direction of the long-term value function.  For all of these reasons, we feel this is a useful example to do in detail. \medskip

We consider a trader in charge of a portfolio of $d \geq 1$ assets, of initial positions $q_0 = (q_0^1,\dots,q_0^d)^{\mathbf{T}}$ where $q_0^i \in \mathbb{R}$ for all $i\in \{1,\dots,d\}$. The trader wants to unwind this portfolio over the time horizon $[0,T]$, where $T>0$. Given a control process $(v_t^i)_{t\in [0,T]}$ representing the trading rate on asset $i$, the inventory process of the $i$-th asset is given by
\begin{align}
    q_t^i = q_0^i - \int_0^t v_s^i ds, \quad i\in \{1,\dots,d\}.
\end{align}
For each stock, we consider Gaussian price dynamics: 
\begin{align*}
    dS_t^i = \sigma^i dW_t^i,
\end{align*}
where the Brownian motions $(W_t^i)_{t\in [0,T]}$ are such that $(S^1_t,\dots,S^d_t)_{t\in [0,T]}$ has a nonsingular covariance matrix $\Sigma$. \medskip 

We treat all temporary impact as instantaneous and permanent impact as linear, hence irrelevant in the continuous-time case. In the single-asset case, if we trade $\Delta q$ dollars in some small time interval of length $\Delta t$, 
and this costs $\lambda \Delta q / \Delta t$ times traded notional for some $\lambda > 0$, 
then the total cost in dollars per unit time is 
\[ 
	\lambda (\delta q / \delta t)^2 \equiv c(\delta q / \delta t), 
\]
where $c(v) = \lambda v^2$. However, $c(\cdot)$ does not need to be quadratic, merely convex. In the multi-asset case, we simply set $c : \mathbb{R}^d \to \mathbb{R}_+$. 
As an aside, we note that these impact assumptions are an approximation that is only ever assumed to be valid within a certain regime. For example, if we repeatedly aggress with medium to large order sizes within a short time-frame, it is unrealistic to assume that the impact will revert instantly. \medskip

The trader attempts to maximize the expectation of utility of final wealth. Due to the nonlinear nature of most utility functions, this is inconvenient to work with. If time is discrete and if the multi-period asset return vector follows an elliptical distribution, then there exists some constant $\kappa > 0$ for which it is equivalent to maximize the mean-variance
quadratic form: 
\[ 
	\mathbb{E}[w_T] - \frac{\kappa}{2} \mathbb{V}[w_T],
\] 
where $w_T$ is the trader's wealth at terminal time. This is essentially the Markowitz prescription. Following \cite{almgren2001optimal}, 
most authors and practitioners replace the variance of final wealth 
$\mathbb{V}[w_T]$ with the integrated instantaneous variance, leading to 
the standard continuous-time approximation of the mean-variance form, 
\[ 
	\int_0^T \Big[  - \frac{\kappa}{2} q_t^T \Sigma q_t - c(\dot q_t) \Big] \, dt. 
\]
We shall now recast the maximization of the mean-variance form of the utility function as a problem in the calculus of variations. 
As in classical mechanics, it has both a Lagrangian and a 
Hamiltonian formulation, which are convex duals to each other. The Hamiltonian is related to the Lagrangian by the Legendre-Fenchel transform. In the following, we define what we mean by the term ``smooth relaxation'' which is not standard terminology.  \medskip

\begin{definition}\label{def:smooth-relaxation}
Let $q_T \in \mathbb{R}^d$ be a desired final portfolio to be achieved at time $T$. 
The \textbf{smooth relaxation problem} associated to $c(\cdot), \Sigma, q_T$ is defined to be:
\begin{align}
V(0,q_0) & =	\min_{q \in C^2([0,T], \mathbb{R}^d) } \int_0^T L( q_s, \dot q_s) \, ds
     \text{ subject to } q_0 = q_0, q_T = q_T,
\label{vari1} 
\end{align}
where the (autonomous) Lagrangian is given by: 
\begin{equation}\label{Lagrangian}
	L(q, v) = c(v) + \frac{1}{2} \kappa (q - q_T)' \Sigma (q - q_T),
\end{equation}
and $\kappa > 0$ is the risk-aversion constant. 
\end{definition}

The terminology of Definition \ref{def:smooth-relaxation} arises because Problem \ref{full-problem} is non-smooth and perhaps 
intractable, given that the action space of Problem \ref{full-problem} is discrete 
and quite large. In this sense \eqref{vari1} is a relaxation of the intractable problem to the space of twice-differentiable paths $C^2([0,T], \mathbb{R}^d)$. 
In this paper, most of our results assume autonomous Lagrangians for simplicity. However, non-autonomous Lagrangians also arise in trading problems. For example, an alpha forecast which attenuates for large $t$ entails a time-dependent linear term in \eqref{vari1}.  \medskip

An application of the Euler-Lagrange formula to \eqref{vari1} leads to a system of second-order equations. By a standard trick, a first-order system can be obtained if we introduce  the so-called generalized momenta $p$, defined as
\begin{equation} \label{eq:def-p}
	p := \partial_v L(q, \dot q).
\end{equation}
If the conditions of the implicit function theorem are satisfied,
we could solve \eqref{eq:def-p} for $\dot q$, obtaining
\[
\dot q = \phi(q, p),
\]
for some function $\phi$ defined implicitly by \eqref{eq:def-p}.
The Euler equation then takes the form
\[
\dot p = \partial_q L(q, \dot{q}) = \partial_q L(q, \phi(q, p) ) \equiv \psi(q, p),
\]
where this defines $\psi$.  \medskip

As the functions $\phi, \psi$ are algebraic (not involving derivatives),
we have a system of $2d$ first-order ODEs given by
\begin{equation}\label{ode1}
    \dot q = \phi(q, p),  \quad
    \dot p = \psi(q, p)
\end{equation}
These equations can be expressed more symmetrically by introducing the
Hamiltonian
\[
    H(q, p) := p \phi(q, p) - L(q, \phi(q,p)).
\]
Equations~\eqref{ode1} are equivalently written in a form known as \emph{Hamilton's equations}:
\begin{equation} \label{eq:hamilton}
	\dot q = H_p(q,p), \quad \dot p = - H_q(q,p).
\end{equation}
Suppose $c(v) = \frac12 v^{\mathbf{T}} \Lambda v$ where $\Lambda = \text{diag}(\lambda^1, \dots, \lambda^d)$ is a diagonal matrix. We assume that no trading is free of cost, so $\lambda^i > 0$ for all $i$. From \eqref{Lagrangian} and \eqref{eq:def-p}, we see that the generalized momenta are $p = \Lambda \cdot q$ and hence algebraically solving, one has $\phi(q,p) = \Lambda^{-1} p$. The Hamiltonian is then 
\[ 
	H(q, p) = \frac{1}{2} p \Lambda^{-1} p - \frac{1}{2} \kappa q^{\mathbf{T}} \Sigma q  .
\]  
Hamilton's equations then become: 
\begin{equation}\label{eq:ham} 
	\dot q = \Lambda^{-1} p, \qquad \dot p = \kappa \Sigma q ,
\end{equation}
and some technical computations lead to the following proposition.
\begin{lemma}
The solution to \eqref{vari1} is given by
\begin{align}\label{sol_AC}
  q_t^{\star} =  (C^{\mathbf{T}})^{-1} \Omega \Big(e^{D^{\frac{1}{2}}(T-t)}-e^{-D^{\frac{1}{2}}(T-t)}\Big)\Big(e^{D^{\frac{1}{2}}T}-e^{-D^{\frac{1}{2}}T}\Big)^{-1} \Omega^{\mathbf{T}}C^{\mathbf{T}} q_0, 
\end{align}
where $\Sigma = C C^{\mathbf{T}}$ is the Cholesky decomposition of $C$ and $\Omega D \Omega^{\mathbf
{T}}$ is a spectral decomposition of the positive definite matrix $\frac{\kappa}{2} C^{\mathbf{T}}\Lambda^{-1}C$. 
\end{lemma}
The trading curve \eqref{sol_AC} can be computed in advance, and corresponds to the order scheduling decision.  We end this section by showing that the gradient of the value function \eqref{vari1} is equal to the generalized momenta $p$. \medskip

\begin{theorem} \label{thm:value-function-gradient}
Let $V:[0,T]\times\mathbb{R}^d \to \mathbb{R}$ be continuously differentiable in time and space such that:
\begin{align}
V(t,q_t) & =	- \min_{q \in C^2([t,T], \mathbb{R}^d) } \int_t^T L( q_s, v_s) \, ds
     \text{ subject to } q_0 = q_0, q_T = q_T
\label{vari2} 
\end{align}
where for all $s\in [0,T]$, $\dot q_s = v_s$ and $L( q, v)$ is separable.\footnote{The value function \eqref{vari1} satisfies these hypothesis, as the associated optimal control \eqref{sol_AC} is continuously differentiable with respect to time and the Lagrangian of the problem is quadratic with respect to both of its variables.} The function $V$ defined in \eqref{vari2} satisfies the Hamilton-Jacobi-Bellman differential equation:
\begin{equation}\label{eq:hj}
    \partial_t V(t,q) + H(q, \nabla V)=0,
\end{equation}
where $H(q,p)= \sup_v \left\{\langle p, v\rangle  -L(q, v) \right\}$,
with the singular final condition:
\[
    V(T,q) = 
    \begin{cases}
    0, \text{ if $q = q_T$}\\
    \infty, \text{ if $q \neq q_T$}.
    \end{cases}
\]
\end{theorem}

\begin{proof} Let $q^\star$ be the path that solves \eqref{vari2} on $[t,T]$ with initial condition $q^\star_t =q$. By the dynamic programming principle, we have for $h>0$
\begin{align}
\label{eq:hj0}
V(t,q) &= -\int_t ^{t+h} L(q^\star_s, \dot q^\star_s) ds + V(t+h, q^\star_{t+h}).
\end{align}
As $ V(t+h, q^\star_{t+h}) = V(t,q)+  \int_t^{t+h}\big(\partial_t V(s,q_s) + \langle \nabla V(s,q_s) \cdot \dot q_s^\star \rangle \big) ds$, Equation \eqref{eq:hj0} can be rewritten as 
\begin{align*}
    0 = \int_t ^{t+h}  \big( \partial_t V(s,q_s) + \langle \nabla V(s,q_s) \cdot \dot q_s^\star - L(q^\star_s, \dot q^\star_s) \big) ds.
\end{align*}
The conclusion follows from an application of the Bellman's optimality principle, see \cite{dreyfus1960dynamic}, which gives the desired Hamilton-Jacobi-Bellman equation. 
\end{proof}

The above theorem gives the desired interpretation of the generalized momenta $p$ in terms of the value function. Indeed, along an optimal trajectory $q^\star$, we have
\[
    \left. \nabla V \right|_{x = q^\star_t} = \partial_v L(q^\star_t, \dot q^\star_t) = p_t 
    \ \forall \ t \in [0,T],
\]
where the second equality is just the definition of $p$ from \eqref{eq:def-p}.

\section{From smooth relaxation to microstructure decision}
\label{sec_smooth_relaxation}

In this section, we prove our main theorem, Theorem \ref{thm_main1}, which shows how a risk-neutral instantaneous-profit maximizer (or ``myopic agent'') can achieve long-term optimality given a judicious choice of microstructure alpha model. In other words, there is a specific microstructure alpha model related to Hamilton's generalized momenta, which, if used by a microstructure trader, encourages the trader to take positions that are optimal at a much longer horizon. \medskip

Suppose now that $L$ is coercive of degree $r > 1$. One may prove
that $H$ coincides with the Fenchel conjugate of $L$:
\begin{equation}\label{eq:ham-fenchel}
	H(q, p) = \sup_{v \in \mathbb{R}^d}  \left\{ \langle p,v \rangle -  L(q,v) \right\}.
\end{equation}
It follows that $H$ is convex in the $p$ variable.
We now restrict attention to autonomous and separable Lagrangians that take the form
\begin{equation} \label{eq:separation}
	L(q,v) = c(v) + f(q),
\end{equation}
which includes the mean-variance example discussed before. Under the assumption \eqref{eq:separation}, duality between the Lagrangian and Hamiltonian implies that the optimal instantaneous trade $\dot q^*(t)$ at each time $t$ is the argument $v$ which solves the maximization problem in \eqref{eq:ham-fenchel}. That is, 
\begin{equation}\label{eq:argmax}
	{\dot q}^\star = \mathop{\text{argmax}}_v \big\{ \langle p, v \rangle - c(v) \big\} .
\end{equation}
where \eqref{eq:separation} implies we can safely omit the term $f(q)$. \medskip

One can interpret \eqref{eq:argmax} as optimization in a risk-neutral world. Indeed, suppose a risk-neutral agent had a vector of expected returns, which happened to exactly equal the vector of generalized momenta, $p$, and sought only to maximize net profit, irrespective of risk. In that case, \eqref{eq:argmax} is the problem faced by this agent. This agent can be considered ``myopic'' because 
any information concerning more than one period ahead is available 
to the agent only indirectly, insofar as $p$ depends on the rest of the trading path. 
This intuition is related to the interpretation of the 
generalized momenta as the gradient of the value function.  
Indeed, for a myopic risk-neutral investor who does not face transaction costs, the ``value function'' of a position of size $x$ is simply the expected profit over the next period, i.e., $V(t,q) = r \cdot q$ where $r$ is the vector of expected returns, but then $\nabla V = r$. The following definition characterizes what a myopic agent is. Moreover, we emphasize that, throughout the paper, $V(\cdot,\cdot)$ corresponds to the value function of a long-term optimization problem (whose gradient provide the effective microstructure alphas), whereas $W(\cdot,\cdot)$ denotes the value function of the myopic agent.

\begin{definition}
A \textbf{myopic agent} with microstructure alphas $p$ 
is defined to be a risk-neutral trader seeking to maximize instantaneous net profit by choosing trading rate given by 
\begin{align}\label{def_pb_myopic_agent}
    {\dot q}^\star = \mathop{\text{argmax}}_v \big\{ \langle p, v \rangle - c(v) \big\},
\end{align}
where $p$ denotes a set of microstructure alphas, as in Definition \ref{def:microstructure-alpha}. The \textbf{value function of the myopic trader} at time $t\in [0,T]$ for an inventory $q\in \mathbb{R}^d$ is defined as
\begin{align}\label{def_value_function_myopic_agent}
    W(t,q) = \int_t^T \big(\langle p_s, v^\star_s \rangle - c(v^\star_s)\big) ds ,
\end{align}
where $v^\star$ is defined as the solution of \eqref{def_pb_myopic_agent}. 
\end{definition}
In other words, the value function of a myopic trader defined in \eqref{def_value_function_myopic_agent} is simply the sum of his instantaneous trading gains over time. The following proposition shows that a myopic trader sending market orders only has to choose $p=\nabla V$ in order to minimize the error between his value function and the long-term objective function $V$. 

\begin{theorem}\label{thm_main1}
Assume that $L$ takes the separable form \eqref{eq:separation}. 
A myopic agent with instantaneous cost function $c(\cdot)$ 
must choose microstructure alphas $p=\nabla V$ in order to minimize the absolute error between his value function and the long-term objective function $V$ defining the trading schedule. More precisely, for all $(t,q)\in[0,T]\times\prod_{i=1}^d\big[\min(q_0^i,q_T^i),\max(q_0^i,q_T^i)\big]$,
\begin{align*}
    \big|W(t,q)-V(t,q)\big| \leq \kappa(T-t)\frac{|q-q_T|^{\mathbf{T}}\Sigma|q-q_T|}{2},
\end{align*}
and we have the uniform bound 
\begin{align*}
    \sup_{(t,q)\in [0,T]\times\prod_{i=1}^d [\min(q_0^i,q_T^i),\max(q_0^i,q_T^i)]} \big|W(t,q)-V(t,q)\big| \leq \kappa T \frac{|q_0-q_T|^{\mathbf{T}}\Sigma|q_0-q_T|}{2},
\end{align*}
where $|q_0-q_T| = \big( |q_0^1-q_T^1|,\dots |q_0^d-q_T^d|\big)$. 
\end{theorem}
\begin{proof}
The myopic trader aims at minimizing $c(v)-p v$, where for the moment $p$ remains undetermined, at each trading time. Over $[t,T]$ the trader's problem can be written:
\begin{align*}
    W(t,q) = \max_v \int_t^T [\langle p_s, v_s\rangle - c(v_s)]\,  ds. 
\end{align*}
We choose quadratic costs $c(v)= \frac{\eta}{2}\|v\|_2^2$, where $\|\cdot \|_2$ is the Euclidian norm and the first order condition with respect to $v$ gives $W(t,q) = \int_t^T \frac{\|p_s\|_2^2}{2\eta} \, ds$. On the other hand, the value function $V(t,q)$ becomes 
\begin{align*}
    \int_t^T \left\{ \frac{1}{2\eta}\|\partial_v L(q^\star_s,v^\star_s)\|_2^2 + \frac{1}{2}\kappa (q^\star_s-q_T)^{\mathbf{T}} \Sigma (q^\star_s-q_T)
    \right\}
    \, ds,
\end{align*}
where $q^\star$ is defined by \eqref{sol_AC} and $v^\star$ is its derivative with respect to time. Therefore, 
\begin{align*}
    | V(t,q)- W(t,q)| = \Big|\int_t^T
    \Big\{
    \frac{1}{2\eta} \|\partial_v L(q^\star_s,v^\star_s)\|_2^2 + \frac{1}{2}\kappa (q^\star_s-q_T)^{\mathbf{T}} \Sigma (q^\star_s - q_T) - \frac{\|p_s\|_2^2}{2\eta} 
    \Big\}
    \, ds \Big|,
\end{align*}
and the minimum with respect to $p$ is attained at $p_s =  \partial_v L(q^\star_s,v^\star_s) = \nabla V(s,q^\star_s)$, because $p$ must depend only of the instantaneous trading rate and $L$ is separable and additive. The bounds are obtained easily by definition of the space of inventories. 
\end{proof}

This simple result has several important consequences. Suppose one wants to avoid the use of optimal control and still wants to follow the Almgren-Chriss trading curve. In that case, one can simply solve the static optimization problem \eqref{eq:argmax} at discrete times (the times of trading), using $p=\nabla V$. Equation \eqref{eq:argmax} does not give a full set of instructions for the trader with a long-term trading schedule who has to interact with a continuous limit order book market, but it can serve as a guide. Indeed, the order routing problem, treated notably in \cite{cont2017optimal}, takes into account the possibility to send limit, market, or cancel orders to several liquidity venues, depending on their spread and imbalance. Stochastic control appears to be inefficient for this problem, as one needs to solve a high-dimensional Hamilton-Jacobi-Bellman equation. Methods involving deep reinforcement learning have been developed for optimal trading, see for example \cite{baldacci2020adaptive}, but they lead to high computation time, especially if one wants to deal with a portfolio of assets traded on several venues. The advantage of the methodology presented in this paper is that one can avoid optimal control and solve a simple static optimization problem to determine the optimal action at each discrete trading time.     \medskip

\begin{remark}
The bounds on the absolute error between the value function of the myopic agent and the long-term objective function enable to compute the accuracy of the myopic trader. For example, take the liquidation over $T=1$ day of $q_0^1=q_0^2=2000$ shares of $2$ assets, with correlation $\rho=0.6$ and daily volatilities $\sigma^1=0.015$ and $\sigma^2 = 0.02$. The absolute error between the two value functions is uniformy bounded in time and inventories by $2\times 10^{-2}$. 
\end{remark}

In Theorem \ref{thm_main1}, the myopic trader does not consider the properties of an order book, such as the possibility to submit limit and market orders or to wait. In this case, the optimal effective microstructure alpha $p$ (in the sense of minimization of the error with respect to the Almgren-Chriss value function $V$) should not be equal to the generalized momenta because we add microstructure effects for the myopic trader that are not present in the trading schedule represented by $\nabla V$. \medskip

With a sufficiently simple fill model, the myopic trader's problem dealing with microstructure effects can be solved in closed form. This closed-form expression (see Example \ref{example:myopic}) illustrates the contrast between the two possible decisions the trader must face, as mentioned above.

\begin{example}\label{example:myopic}
For the sake of readability, we assume $d=1$. Suppose that the myopic trader can choose between submitting a limit order (with fill probability $0<f<1$) or a market order (with the cost of crossing the spread equal to $s>0$). 
The myopic trader's optimization problem is:
\begin{align*}
    \sup_v \Big\{ (pv -\frac{\eta}{2}v^2 - sv)\mathbf{1}_{\{pv -\frac{\eta}{2}v^2 - sv > f (pv - \frac{\eta}{2}v^2) \}} + f (pv - \frac{\eta}{2}v^2)\mathbf{1}_{\{pv -\frac{\eta}{2}v^2 - sv < f (pv - \frac{\eta}{2}v^2) \}} \Big\}.
\end{align*}
Computations lead to the following decisions:
\begin{align*}
    \text{ passive order if } p<\frac{2s}{1-f}, \text{ aggressive otherwise},
\end{align*}
and the optimal microstructure alpha is given by: 
\begin{align*}
    p^\star_t = \frac{1}{\sqrt{f}}|\nabla V(t,q^\star_t)| \text{sgn}\big(\nabla V(t,q^\star_t)\big) \text{ if } p_t^\star <\frac{2s}{1-f}, \quad p^\star_t =  |\nabla V(t,q^\star_t)|\text{sgn}\big(\nabla V(t,q^\star_t)\big)+s \text{ otherwise.}
\end{align*}
\end{example}

The use of the generalized momenta as effective microstructure alpha has a wide range of practical implications. First, it offers a way to bridge the gap between order placement decisions and scheduling decisions, usually decoupled in practice. Second, the microstructure formulation helps to tackle classic optimal control on limit order books. For example, a realistic optimal trading framework dealing with a portfolio of assets on several liquidity venues is in practice intractable due to the dimensionality of the problem. In the method presented in this paper, the optimal controls of the trader (that is, the volume sent on each venue for each asset by the mean of limit and market orders) are derived through a simple static optimization problem, which can be solved for a large number of assets on a large number of venues. The convergence through the trading schedule is guaranteed by choice of the effective microstructure alpha $p$. 
\begin{remark}
Note that, in the framework of Example \eqref{example:myopic}, if the spread $s$ tends to zero and the filling probability $f$ tends to one, we recover the framework of a myopic sending market orders only, and the optimal effective microstructure alpha is given by $\nabla V$. 
\end{remark}
This method can easily handle the increasing complexity coming from the microstructure effects (short term alpha, imbalance, and spread of each venue, etc.). In the next section, we present a general microstructure trading model taking into account the main stylized facts combining order placement and order routing of a portfolio of assets. We show that using the method proposed in this paper can be applied to solve in practice two important problems in systematic trading: the multi-asset, multi-venue optimal trading problem and the multi-asset, multi-venue optimal market-making problem.

\section{A general microstructure trading model with long-term trading schedule}
\label{sec_general_microstructure}

\subsection{Optimal trading}\label{subsec_opt_trading}

We first provide a definition of the problem.
\begin{definition}
Consider an agent trading a portfolio of correlated assets, where each asset is listed on one or more liquidity venues. The \textbf{multi-asset, multi-venue optimal trading problem} consists in determining at a given time and for each asset, the optimal quantity to buy or sell on each venue, for given market conditions and a pre-computed trading schedule, as well as the optimal limit at which such quantity should be posted. 
\end{definition}

The framework described here is inspired by \cite{baldacci2020adaptive}. Consider a trading schedule for $d\geq 1$ assets $q^\star \in \mathbb{R}^d$ (the Almgren-Chriss trading schedule described in Equation \eqref{sol_AC}, for example) with associated value function $V(t,q_t^\star)$. For each asset $i\in \{1,\dots,d\}$, the trader splits his limit and market orders between $N^i\geq 1$ liquidity venues. We assume that he wants to unwind the portfolio so that $q_T = 0^d$. For all $i\in \{1,\dots,d\},n\in \{1,\dots,N^i\}$, the order book of the asset $i$ on the venue $n$ is characterized by the following quantities:
\begin{itemize}
  \item the bid-ask spread process $(\psi_t^{i,n})_{t \in [0,T]}$ taking values in the state space $\overline{\psi}^{i,n}=\{\delta^{i,n},\dots,J\delta^{i,n}\}$,
  \item the imbalance process $(I^{i,n}_t)_{t \in [0,T]}$ taking values in the state space $\overline{I}^{i,n}=\{I^{i,n}_1,\dots,I^{i,n}_K\}$,
\end{itemize}
where $J, K \in \mathbb{N}$ denote the number of possible spreads and imbalances respectively and $\delta^{i,n}$ stands for the tick size of $i$-th asset on the $n$-th venue. Note that the dynamics are unspecified, meaning that any continuous-time stochastic process with discrete values can be considered for the purpose of simulation.
\begin{definition}
A \textbf{market regime} is, for asset $i\in \{1,\dots,d\}$, a set of spread and imbalance values on the different venues $n\in \{1,\dots,N^i\}$.
\end{definition}
We define the sets $\Psi=\{\Psi_1,\dots,\Psi_{\# \Psi}\}, \mathcal{I}=\{\mathcal{I}_1,\dots,\mathcal{I}_{\# \mathcal{I}}\}$ of disjoint intervals, representing different market regimes of interest in terms of spreads and imbalances.
\begin{example}
Assume $d=1$ and for all $n\in \{1,\dots,N\}$,  $\delta^n=\delta$. The set $\Psi=\big\{\delta, \{2\delta, 3\delta\}, \{4\delta, 5\delta\} \big\}$ denotes three spread regimes: low (one tick), medium (two or three ticks), and high (four or five ticks). 
\end{example}
\begin{example}
Assume $d=1$ and for all $n\in \{1,\dots,N\}$ and $k\in \{1,\dots,K\}$ that $I^n_k=I_k$. In this case the set $\mathcal{I}=\big\{[-1,-0.66], (-0.66, -0.33], (-0.33, 0.33], (0.33, 0.66], (0.66, 1]\big\}$ denotes five regimes of imbalance: low ($-33\%$ to $33\%$), medium on the ask (resp. bid) from $33\%$ to $66\%$ (resp. from $-66\%$ to $-33\%$) and high on the ask (resp. bid) from $66\%$ to $100\%$ (resp. from $-100\%$ to $-66\%$).
\end{example}

The number of, possibly partially, filled ask orders on the asset $i$ in the venue $n$ is modeled by a Cox process denoted by $N^{i,n}$ with intensities $\lambda^{i,n}\big(\psi^i_t, I^i_t, p_t^{i,n},\ell^i_t\big)$ where $p_t^{i,n}\in Q_\psi^{i,n}$ represent the limit at which the trader sends a limit order of size $\ell_t^{i,n}$, and
\begin{align*}
 & Q_\psi^{i,n}=\{0,1\} \text{ if } \psi^{i,n}=\delta^{i,n},\text{ and } \{-1,0,1\} \text{ otherwise.}
\end{align*}
Practically, on asset $i$, for $n\in \{1,\dots,N^i\}$, when the spread is equal to the tick size, the trader can post at the first best limit ($p^{i,n}=0$) or the second best limit (if $p^{i,n}=1$). When the spread is equal to two ticks or more, the trader can either create a new best limit ($p^{i,n}=-1$) or post at the best or the second best limit as previously. The arrival intensity of a buy market order at time $t$ on the venue $n\in \{1,\dots,N^i\}$ for asset $i$ at the limit $p\in Q_\psi^{i,n}$, given a couple $(\psi^i_t,I^i_t)=\mathbf{m}$ of spread and imbalance on each venue, is equal to $\lambda^{i,n,\mathbf{m},p}>0$. When the trader posts limit orders of volume~$\ell_t^{i,n}$ on the $n$-th venue for $n\in \{1,\dots,N^i\}$, the probability that it is executed is equal to $f^{\lambda}(\ell^i_t)$, where $f^\lambda(\cdot)\in [0,1]$ is a continuously differentiable function, decreasing with respect to each of its coordinate. Therefore, the arrival intensity of an ask market order filling the buy limit order of the trader for asset $i$ on the $n$-th venue at the limit $p_t^{i,n}$, given spread and imbalance $(\psi^i_t,I^i_t)$ is a multi-regime function defined by
\begin{align}\label{eq_lambdas_optex}
  & \lambda^{i,n}(\psi^i_t,I^i_t,p_t^{i,n},\ell^i_t)=f^{\lambda}(\ell^i_t)\sum_{\mathbf{m}\in\mathcal{M}^i,p\in Q_\psi^{i,n}} \lambda^{i,n,\mathbf{m},p}\mathbf{1}_{\{(\psi^i_t,I^i_t)\in \mathbf{m}, p_t^{i,n}=p\}},
\end{align}
where $\mathcal{M}^i=\Psi^{N^i} \times \mathcal{I}^{N^i}$. Moreover, we allow for partial execution, the fact of which we represent by random variables $\epsilon^{i,n}_t \in [0,1]$. The proportion of executed volume for limit orders in each venue depends on the spread and the imbalance in all $N^i$ venues for asset $i$, as well as the volume and the limit of the order chosen by the trader. We assume a categorical distribution with $R>0$ different execution proportions $\omega^r, r\in\{1, \ldots, R\}$ for each venue with $\mathbb{P}(\epsilon_t^{i,n} = \omega^r) = \rho^{i,n,r}(\psi^i_t, I^i_t, p_t^{i,n}, \ell^i_t)$, where
\begin{align}\label{eq_rho_optex}
  \rho^{i,n,r}(\psi^i_t,I^i_t,p_t^{i,n},\ell^i_t) = f^\rho (\ell^i_t)\sum_{\mathbf{m}\in \mathcal{M}^i,p\in Q_\psi^{i,n}} \rho^{i,n,\mathbf{m},p,r}\mathbf{1}_{\{(\psi^i_t,I^i_t)\in \mathbf{m},p_t^{i,n}=p\}}, 
\end{align}
where $f^\rho(\cdot)$ is a continuously differentiable function, decreasing with respect to each of its coordinate. \medskip

We allow for the execution of market orders (denoted by a point process $(J_t^{i,n})_{t\in[0,T]}$) on each venue of size $(v_t^{i,n})_{t\in[0,T]}\in [0,\overline{v}]$ where $\overline{v}>0$ and $J_t^{i,n} = J_{t^-}^{i,n} + 1$. We assume that market orders are always fully executed but this assumption can be relaxed easily. As each asset must be bought or sold, we define $\Delta=(\Delta^1,\dots,\Delta^d)$ where for $i\in \{1,\dots,d\}$, $\Delta^i = 1$ if $q_0^i >0, -1$ otherwise. The inventory process on each asset is defined by 
\begin{align*}
    q_t^i = q_0^i - \Delta^i\sum_{n=1}^{N^i}\Big(\int_0^t \ell_s^{i,n} \epsilon_s^{i,n} dN_s^{i,n} + \int_0^t v_s^{i,n} dJ_s^{i,n}\Big).
\end{align*}

The myopic trader has an effective microstructure $(p_t^{\text{eff},i})_{t\in [0,T]}$ in order to follow the pre-computed execution curve $q_t^{\star i}$ on each asset, but also a short-term alpha $(p_t^{\text{short},i})_{t\in [0,T]}$ which is a function of the current spread and imbalance $\mathbf{m}^i$ of all the venues where asset $i$ is listed.
\begin{remark}
The microstructure alpha considered for each asset is the sum of a direct microstructure alpha depending on the market regimes and an effective microstructure alpha that gives a signal to follow the long-term objective function. The sum of these two terms gives the magnitude of the buy or sell signal. For example, suppose $p^{\text{eff},i}$ is small and positive, indicating that filling  a buy order would be a slight improvement to the value function. Suppose with $p^{\text{eff},i}$ alone, the system would have recommended a passive buy order. Now suppose a strongly-positive microstructure alpha, denoted $p^{\text{short},i}$, is also present; 
then the combination $p^{\text{eff},i} + p^{\text{short},i}$  in place of $p^{\text{eff},i}$ should recommend a more aggressive action, such as a spread-crossing buy order.  
\end{remark}

Finally, the cost function of a limit order of size $\ell$ at limit $p$ on venue $n$ for asset $i$ is defined as $c^{i,n,\mathcal{L}}(\ell,p)$ and $c^{i,n,\mathcal{M}}(v)$ for the cost function of a market order of size $v$ on venue $n$ for asset $i$. \medskip

The myopic trader acts at discrete times and at time $t\in [0,T]$ for $(\psi^i_t,I^i_t) \in \mathbf{m}^i, i\in \{1,\dots,d\}$, his optimization problem is 
\begin{align}\label{opt_pb_myopic}
\begin{split}
    \max\bigg\{ & \sup_{\ell,p} \Big\{ \sum_{i=1}^d \sum_{n=1}^{N^i}   \lambda^{i,n}(\mathbf{m}^i,p^{i,n},\ell^i)\mathbb{E}\Big[ \big( p^{\text{eff},i}+ p^{\text{short},i,n}(\mathbf{m}^i)\big) \ell^{i,n} \epsilon^{i,n}  - c^{i,n,\mathcal{L}}(\epsilon^{i,n}\ell^{i,n},p^{i,n}) \Big]     \Big\},\\
   & \sup_{v} \Big\{\sum_{i=1}^d\sum_{n=1}^{N^i} \big( p^{\text{eff},i}+ p^{\text{short},i,n}(\mathbf{m}^i)\big) v^{i,n}  - c^{i,n,\mathcal{M}}(v^{i,n})\Big\} \bigg\},
\end{split}
\end{align}
where the expectation is taken with respect to the variables $\epsilon^{i,n}$ for all $i\in \{1,\dots,d\}, n\in \{1,\dots,N^i\}$. This is a simple static optimization which can be solved for a large number of assets and venues using a multidimensional root-finding method. The output is, for each state $\mathbf{m}^i$, the optimal volumes and limits $\ell^{\star i,n}(\mathbf{m}^i),p^{\star i,n}(\mathbf{m}^i)$ for each asset on each liquidity venue. We define the value function of the myopic trader at time $t$ as
\begin{align*}
    W(t,\psi,I,q) = \int_t^T  \max \bigg\{ & \sum_{i=1}^d \sum_{n=1}^{N^i}   \lambda^{i,n}(\psi_s^i,I_s^i,p_s^{\star i,n},\ell_s^{\star i})\mathbb{E}\Big[ \big( p_s^{\text{long},i}+ p_s^{\text{short},i,n}(\psi_s^i,I_s^i)\big) \ell_s^{\star i,n} \epsilon_s^{i,n} \\ 
    & - c^{i,n,\mathcal{L}}(\epsilon_s^{i,n}\ell_s^{\star i,n},p_s^{\star i,n}) \Big],
    \sum_{i=1}^d\sum_{n=1}^{N^i} \big( p_s^{\text{long},i}+ p_s^{\text{short},i,n}(\psi_s^i,I_s^i)\big) v_s^{\star i,n}  - c^{i,n,\mathcal{M}}(v^{\star i,n}) \bigg\} ds
\end{align*}
where $\psi_t = \psi, I_t=I, q_t = q$.  As in Theorem \ref{thm_main1}, the myopic trader has now to choose the long-term alpha $p^{\text{eff}}$ to match the trading schedule $q^\star$. This leads to the following  optimization setting for all $\mathbf{m}^i\in \mathcal{M}^i$:

\begin{align}\label{opt_pb_myopic_opt}
\tag{Opt-Trd}
\begin{split}
    \max\bigg\{ & \sup_{\ell,p} \Big\{ \sum_{i=1}^d \sum_{n=1}^{N^i}   \lambda^{i,n}(\mathbf{m}^i,p^{i,n},\ell^i)\mathbb{E}\Big[ \big( p^{\text{eff},i}+ p^{\text{short},i,n}(\mathbf{m}^i)\big) \ell^{i,n} \epsilon^{i,n}  - c^{i,n,\mathcal{L}}(\epsilon^{i,n}\ell^{i,n},p^{i,n}) \Big]     \Big\},\\
   & \sup_{v} \Big\{\sum_{i=1}^d\sum_{n=1}^{N^i} \big( p^{\text{eff},i}+ p^{\text{short},i,n}(\mathbf{m}^i)\big) v^{i,n}  - c^{i,n,\mathcal{M}}(v^{i,n})\Big\} \bigg\}, \\
   & p^{\text{eff}} = \text{argmin}_{p^{\text{eff}}}\big|V(\cdot,q_{\cdot}^\star)-W(\cdot,\psi_{\cdot},I_{\cdot},q_{\cdot})\big|.
\end{split}
\end{align}
In this general framework, order scheduling with a long-term target is easily tractable even for a large portfolio of assets, as the trader has to solve a static optimization problem at each trading time. For a parsimonious model of filling probabilities, the effective microstructure alpha can be computed in closed form. Note that each time a fill is received that changes the portfolio holdings, and/or each time a significant amount of time passes, the effective microstructure alpha $p^{\text{eff}}$ must be recomputed. 

\begin{remark}
The methodology presented in this paper leads to entirely tractable optimization problems, even for a large number of assets. This is the case when we have a closed-form solution for the long-term value function $V$, which can be computed quickly. It also suggests an approximation of the effective microstructure alpha, that is to take $p_t^{\text{eff},i} = \nabla_i V(t,q_{\cdot}^\star)$. This heuristics will be used in the next section to solve a different control problem.
\end{remark}

\subsection{Market-making}

The great advantage of the framework presented in this paper is that it avoids the use of optimal control to tackle optimal trading problems. The trader solves a simple static optimization problem, and the use of the generalized momenta as a long-term alpha plays the role of the trading schedule. Similar ideas can be applied to the market-making problem, with some minor changes. \medskip

\begin{definition}
Consider an agent trading on a portfolio of correlated assets, where each of them is listed on one or several liquidity venues. His goal is to earn the difference between the bid and ask prices (the bid-ask spread) while keeping his inventory close to zero to avoid an unwanted large exposure and be forced to buy at a higher price or sell at a lower price in order to unwind this position. The \textbf{multi-asset, multi-venue optimal market-making problem} consists in deriving at a given time, for each asset, the optimal quantity to buy or sell in each venue, for given market conditions, as well as the optimal limit at which such quantity should be posted, with an inventory vector mean-reverting around zero or some predetermined target. 
\end{definition}

The market-making problem has been introduced in the financial literature by \cite{ho1981optimal,glosten1985bid}. Ho and Stoll presented a framework to tackle inventory management, while Grossman and Miller proposed a $3$ periods model that encompassed both market-makers and final customers, enabled them to understand what happens at equilibrium, and contributed to the important literature on the price formation process. The seminal reference of the recent literature on market-making is the work of Avellaneda and Stoikov in \cite{avellaneda2008high}, who proposed a stochastic control framework to tackle the quoting and inventory management problems. Since then, a vast literature on optimal market-making has emerged, basically adding many features to the Avellaneda and Stoikov framework, see for example \cite{cartea2014buy,gueant2013dealing} and the two textbooks \cite{cartea2015algorithmic,gueant2016financial}. These works deal with single asset market-making, and the considered framework is more suitable for OTC markets rather than order-driven markets. The problem of multi-asset market-making, dealing with the curse of dimensionality, has been addressed via deep reinforcement learning methods, see for example \cite{gueant2019deep}. Models for optimal market-making in limit order books have been developed for the single asset case, see \cite{guilbaud2013optimal}, for example. All these models suffer from the same problem when dealing with a portfolio of assets: solving a high-dimensional Hamilton-Jacobi-Bellman equation makes the problem almost intractable in practice. In this section, we propose an adaptation of the previously described heuristics to tackle the multi-asset market-making problem in limit order books.

\subsubsection{The long-term objective function}

Our methodology to solve optimal control problems in high dimension relies on the fact that the effective microstructure alphas come from a long-term objective function computed analytically. This is the case of the Almgren-Chriss trading curve, which hides the microstructure effects that are incorporated in the myopic optimization problem. However, as stated previously, the main constraint of the optimal market-making problem is that, even for market-making on OTC markets, the value function's computation is very time-consuming. We propose to use the gradient of an approximation of the value function of the optimal market-making problem on OTC markets as the effective microstructure alphas for the optimal market-making problem in order books. To this end, we borrow the OTC framework of \cite{bergault2020closedform} and recall their modeling assumptions briefly.\medskip

For $i\in \{1,\dots,d\}$, the reference price of asset $i$ is modeled by a process $S_t^i$ with dynamics
\begin{align*}
dS_t^i = \sigma^i dW_t^i,
\end{align*}
where $(W_t^i,\dots,W_t^d)$ is a $d$-dimensional Brownian motion with variance-covariance matrix $\Sigma$. At each $t\in [0,T]$, the market-maker chooses the prices $P_t^{i,b},P_t^{i,a}$ at which she is ready to buy/sell each asset $i$. These prices are given by
\begin{align*}
P_t^{i,b} = S_t^i - \delta_t^{i,b}, \quad P_t^{i,a} = S_t^i + \delta_t^{i,a}, 
\end{align*}
where $\delta_t=(\delta_t^{1,b},\delta_t^{1,a},\dots,\delta_t^{d,b},\delta_t^{d,a})$ are the control processes of the market-maker corresponding to the bid and ask spreads set on each asset $i$. For $i\in \{1,\dots,d\}$, the point processes $N_t^{i,b},N_t^{i,a}$ denote the total number of bid and ask transactions between $0$ and $t$ on asset $i$. Their intensities are given by $\Lambda^{i,b}(\delta_t^{i,b}),\Lambda^{i,a}(\delta_t^{i,a})$ where the functions $\Lambda^{i,b},\Lambda^{i,a}$ satisfy some technical conditions, see \cite{bergault2020closedform} for details. These conditions are sufficiently general to allow the use of several form of intensity such as exponential, logistic, SU Johnson etc. \medskip

The transaction size for asset $i$ is constant and denoted by $z^i$, and the inventory process of the market-maker for asset  $i$ is 
\begin{align*}
dq_t^i = z^i \big( dN_t^{i,b} - dN_t^{i,a}\big), \quad q_t=(q_t^1,\dots,q_t^d)^{\mathbf{T}}. 
\end{align*}
The cash process of the market-maker has the following dynamics:
\begin{align*}
dX_t = \sum_{i=1}^d \big(P_t^{i,a}dN_t^{i,a} -  P_t^{i,b}dN_t^{i,b}\big). 
\end{align*}
The optimization problem of the market-maker is defined by
\begin{align*}
    \sup_{\delta} \mathbb{E}\Big[X_T + \sum_{i=1}^d q_T^i S_T^i - \frac{\gamma}{2}\int_0^T q_s^{\mathbf{T}}\Sigma q_s dt \Big],
\end{align*}
and simple computations\footnote{It can be shown by simple change of variables that the value function of this control problem is only a functional of the time and the inventories.} give the associated value function at time $t$ for a given inventory vector $q_t=q$: 
\begin{align}\label{opt_MM}
V(t,q)=\sup_{\delta} \mathbb{E}_t\bigg[\sum_{i=1}^d\int_t^T \Big( \delta_s^{i,a}\Lambda^{i,a}(\delta_s^{i,a}) + \delta_s^{i,b}\Lambda^{i,b}(\delta_s^{i,b})- \frac{\gamma}{2}q_s^{\mathbf{T}}\Sigma q_s\Big) dt \bigg],
\end{align}
where $\mathbb{E}_t$ denote the conditional expectation with respect to the canonical filtration at time $t$ and $\gamma>0$ is the risk-aversion of the market-maker. He wishes to maximize the sum of his cash process and the mark-to-market value of his inventory. The running penalty forces him to mean-revert his inventories to zero. We now state the main proposition of \cite{bergault2020closedform} that provides a closed form approximation of $V(t,q)$, and refer to this article for the proof. 
\begin{proposition}
Define the functions 
\begin{align*}
H^{i,b}(p)=\sup_\delta \{\Lambda^{i,b}(\delta)(\delta-p)\} ,\quad  H^{i,a}(p)=\sup_\delta \{\Lambda^{i,a}(\delta)(\delta-p)\}, 
\end{align*}
and the constants $\alpha_j^{i,b} =(H^{i,b})^j(0),\alpha_j^{i,a} =(H^{i,a})^j(0)$, where the superscript $j\in \{0,1,2\}$ denote the derivative of order $j$. Define also for $k\in \mathbb{N}$
\begin{align*}
& \Delta_{j,k}^{i,b} = \alpha_j^{i,b} (z^i)^k, \quad  \Delta_{j,k}^{i,a} = \alpha_j^{i,a} (z^i)^k, \\
& V_{j,k}^b = \Big(\Delta_{j,k}^{1,b},\dots, \Delta_{j,k}^{d,b}\Big), \quad V_{j,k}^a = \Big(\Delta_{j,k}^{1,a},\dots, \Delta_{j,k}^{d,a}\Big), \\
& D_{j,k}^b = \text{diag}\Big(\Delta_{j,k}^{1,b},\dots, \Delta_{j,k}^{d,b}\Big), \quad D_{j,k}^a = \text{diag}\Big(\Delta_{j,k}^{1,a},\dots, \Delta_{j,k}^{d,a}\Big). 
\end{align*}
Then if $\alpha^{i,b}_2+\alpha^{i,a}_2>0$, the value function of the optimal control problem \eqref{opt_MM} can be approximated by the function 
\begin{align}\label{approx_MM_fct}
  \tilde{V}(t,q) = -q^{\mathbf{T}}A(t) q - q^{\mathbf{T}} B(t) - C(t), 
\end{align}
where $A: [0,T] \to \mathcal{S}_d^{++}, B: [0,T] \to \mathbb{R}^d$ and $C: [0,T] \to \mathbb{R}$ are deterministic functions given by
\begin{align*}
& A(t)= \frac{1}{2}D_+^{-\frac{1}{2}}\hat{A}\big(e^{\hat{A}(T-t)}-e^{-\hat{A}(T-t)}\big)\big(e^{\hat{A}(T-t)}+e^{-\hat{A}(T-t)}\big)^{-1}D_+^{-\frac{1}{2}}, \\
& B(t) = -2e^{-2\int_t^T A(u)D_+ du} \int_t^T e^{2\int_s^T A(u)D_+ du} A(s)\Big(V_- + D_- \mathcal{D}\big(A(s)\big)\Big)ds \\
& C(t)= -\text{Tr}(D_{0,1}^b + _{0,1}^a)(T-t) -\text{Tr}\Big((D_{1,2}^b+D_{1,2}^a)\int_t^T A(s)ds\Big) - V_-^{\mathbf{T}}\int_t^T B(s)ds \\
& -\frac{1}{2}\int_t^T \mathcal{D}\big(A(s)\big)^{\mathbf{T}}(D_{2,3}^b + D_{2,3}^a)\mathcal{D}\big(A(s)\big) ds - \frac{1}{2}\int_t^T B(s)^{\mathbf{T}}D_+ B(s) ds \\
& -\int_t^T B(s)^{\mathbf{T}}D_- B(s) ds,
\end{align*}
with
\begin{align*}
D_+ = D_{2,1}^b+  D_{2,1}^a, D_- = D_{2,2}^b -  D_{2,2}^a, V_- = V_{1,1}^b-V_{1,1}^a, \hat{A} = \sqrt{\gamma} \big(D_+^{\frac{1}{2}} \Sigma  D_+^{\frac{1}{2}}\big)^{\frac{1}{2}},
\end{align*}
and $\mathcal{D}$ is the linear operator mapping a matrix onto the vector of its diagonal and $\mathcal{S}_d^{++}$ is the set of $d\times d$ definite positive matrix. 
\end{proposition}

The approximated value function \eqref{approx_MM_fct} is quadratic, therefore, sub-differentiable with respect to the vector of inventories $q$ and the deterministic functions $A(t),B(t),C(t)$ can be computed in closed form. It takes into account the main property of a ``high-level'' multi-asset market-making problem, that is, the correlation structure between the assets. By analogy with Section \ref{subsec_opt_trading}, its sub-gradient can be chosen as an effective microstructure alpha to mean revert toward a flat inventory.
\begin{remark}
In order to compute efficiently the value function in \eqref{approx_MM_fct}, note that the matrix $A$ can be diagonalized and therefore approximated with a principal component analysis. The expressions of $A,B,C$ do not provide intuition about the long-term behavior of the market-maker. However, in an asymptotic framework, that is when $T\to +\infty$, we obtain 
\begin{align*}
    & A \to_{T\to+\infty} \frac{1}{2}\sqrt{\gamma}\Gamma, \\
    & B \to_{T\to +\infty} -D_+^{-\frac{1}{2}}\hat{A}\hat{A}^+D_+^{-\frac{1}{2}}\big(V_- + \frac{1}{2}\sqrt{\gamma} D_- \mathcal{D}(\Gamma) \big),
\end{align*}
where $\Gamma=D_+^{-\frac{1}{2}}\big(D_+^{\frac{1}{2}}\Sigma D_+^{\frac{1}{2}}\big)^{\frac{1}{2}}D_+^{-\frac{1}{2}}$ and $\hat{A}^+$ is the Moore-Penrose generalized inverse of $\hat{A}$. If we perform a principal component analysis on the variance-covariance matrix $\Sigma$, we observe that the buy or sell signal (depending on the sign of the inventories) coming from the sub-gradient of $V(t,q)$ is an increasing function of the eigenvalues of $\Sigma$ and the risk-aversion parameter $\gamma$. Thus, choosing the sub-gradient of $V(t,q)$ as an effective microstructure alpha should provide a mean-reverting signal for a myopic agent, taking into account the correlation between the assets.
\end{remark}

\subsubsection{Multi-asset multi-venue optimal market-making in limit order book}

We now derive the solution to the multi-asset multi-venue optimal market-making problem using our heuristics. We take the same modeling notation as in \ref{subsec_opt_trading}. Assuming bid and ask symmetry for sake of simplicity, we introduce the processes $N^{i,n,b},N^{i,n,a}$ to model the number of (possibly partially-filled) bid and ask orders on the asset $i$ in the venue $n$ of intensity $\lambda^{i,n}(\psi^i,I^i,p^{i,n,b},\ell^{i,b})$ for the bid side and $\lambda^{i,n}(\psi^i,I^i,p^{i,n,a},\ell^{i,a})$ for the ask side, where the function $\lambda$ is defined by \eqref{eq_lambdas_optex}. The quantity $\ell^{i,n,b}$ (resp. $\ell^{i,n,a}$) is the volume sent on the $n$-th venue of the $i$-th asset on the bid (resp. ask) side. The quantity $p^{i,n,b}$ (resp. $p^{i,n,a}$) is the limit chosen by the market-maker on the $n$-th venue of the $i$-th asset on the bid (resp. ask) side to send a limit order. The distribution of the random variables $\epsilon_t^{i,n,b}$ and $\epsilon_t^{i,n,a}$ are defined as in Equation \eqref{eq_rho_optex}. The market-maker can also send market orders on the bid and ask sides (denoted by point processes $J_t^{i,n,b},J_t^{i,n,b}$) on each venue of size $v_t^{i,n,b},v_t^{i,n,a}$. Its inventory process on each asset is defined by
\begin{align*}
    q_t^i = \Big(\int_0^t \ell_s^{i,n,b} \epsilon_s^{i,n,b} dN_s^{i,n,b} + \int_0^t v_s^{i,n,b} dJ_s^{i,n,b}\Big) - \Big(\int_0^t \ell_s^{i,n,b} \epsilon_s^{i,n,b} dN_s^{i,n,b} + \int_0^t v_s^{i,n,b}  dJ_s^{i,n,b}\Big) \in \mathbb{Z}. 
\end{align*}
The problem faced by a market-maker is slightly different compared to a classic trader. While the trader must follow a predetermined target, the market-maker's inventory must revert toward zero. Therefore, we seek a long-term alpha that gives a signal to our myopic market-maker of the form ``sell for high inventory, buy for low inventory'' with different type of aggressiveness (limit or market order) depending on the level of inventory. Contrary to optimal execution, there is no optimal inventory in market-making problems at a given time $t\in [0,T]$, which explains the dependence of the long-term alpha on the current inventory. \medskip

As stated previously, the effective microstructure alphas should be the gradient of a value function corresponding to a ``high-level'' multi-asset market-making problem (which hides the microstructure effects). This value function should be in closed form to recompute the gradient quickly when the market-maker trades and too much time passed. Thus, we propose to use the sub-gradient of the value function \eqref{approx_MM_fct} corresponding to an approximation of the multi-asset market-making value function in OTC markets as a proxy for effective microstructure alphas used by a market-maker acting on a portfolio of assets listed on several order book platforms.
By analogy with Section \ref{subsec_opt_trading}, at each time step $t$, for an inventory vector $q\in \mathbb{R}^d$, the market-maker solves the following optimization problem:
\begin{align}\label{opt_pb_myopic_MM}
    \max\bigg\{\! & \sup_{\ell,p} \!\Big\{\! \sum_{i=1}^d \sum_{n=1}^{N^i} \! \sum_{j\in \{b,a\}} \!\!\!\! \lambda^{i,n}(\mathbf{m}^i,\!p^{i,n,j},\!\ell^{i,j})\mathbb{E}\Big[ \big( p^{\text{eff},i}\!+\! p^{\text{short},i,n}(\mathbf{m}^i)\big) \ell^{i,n,j} \epsilon^{i,n,j} \! -\! c^{i,\!n,\!\mathcal{L}}(\epsilon^{i,n,j}\ell^{i,n,j},\!p^{i,n,j}) \Big]     \Big\},\nonumber\\ \tag{Opt-MM}
   & \sup_{v^b} \Big\{\sum_{i=1}^d\sum_{n=1}^{N^i} \big( p^{\text{eff},i}+ p^{\text{short},i,n}(\mathbf{m}^i)\big) v^{i,n,b}  - c^{i,n,\mathcal{M}}(v^{i,n,b})\Big\} , \nonumber\\
   & \sup_{v^a} \Big\{\sum_{i=1}^d\sum_{n=1}^{N^i} \big( p^{\text{eff},i}+ p^{\text{short},i,n}(\mathbf{m}^i)\big) v^{i,n,a}  - c^{i,n,\mathcal{M}}(v^{i,n,a})\Big\} \bigg\}, \nonumber \\
   & p^{\text{eff},i} = \nabla^i \tilde{V} (t,q).\nonumber
\end{align}
The control problem is essentially a choice between sending limit orders or market orders in each venue for each asset. The effective microstructure alpha helps the market-maker to mean revert his inventory toward zero. For example, assume that the market-maker received a large buy passive filling in asset $i$. The effective microstructure alpha, that is the $i$-th component of the gradient of the long-term utility function $\tilde{V}$, will point down which is a strong sell signal. Therefore, the market-maker will send a sell market order to reduce his long position. Note that the effective microstructure alpha takes into account the correlation structure between the assets, meaning that the market-maker can hedge a long position in an asset with a short position in another positively correlated asset. \medskip

\section{Numerical results}\label{sec_numerical_results}

Mathematical elegance and simplicity are to be prized, of course, but an execution model cannot pass the test of practicality until it helps us execute portfolio transitions.  \medskip

One of the most important features of our framework, as compared
with a plain-vanilla, Almgren-Chriss executor, it allows the executor to consider market microstructure and use passive orders, hence avoiding certain types of market impact and spread costs. The main point we wish to make in this example is that our method potentially avoids the pitfalls of a purely-passive execution model because it can consider the utility gradient (and its multi-period analog, the gradient of the value function) in the formation of aggression levels. With this specific aim in mind, we consider the liquidation of a market-neutral portfolio with our 
method and contrast this with comparable results for a purely-passive 
method. \medskip

The specific example we choose is the liquidation of a market-neutral portfolio on October $15, 2008$. The portfolio to be liquidated is long IBM and short AAPL. We choose the long position in IBM arbitrarily to be 1000 shares. We estimate the CAPM beta of each security, denoted $\hat{\beta}_i$ where $i = 1,2$, using three years of daily data, and size the short position so that the beta exposure of the portfolio $\sum_i h_i \hat{\beta}_i$ is near zero. \medskip

\subsection{Transaction cost model and microstructure simulation}

We take $\kappa = 10^{-3}$ and we assume $c(v) = v^{\mathbf{T}} \Lambda v$ for some diagonal 
matrix $\Lambda = \text{diag}(\lambda^1, \ldots, \lambda^d)$. This reduces 
the transaction cost modeling problem to one of estimating appropriate 
values for each $\lambda_i$. Let $\text{advp}_i$ denote our prediction of the daily dollar volume in the $i$-th security. The notation ``advp'' comes from the fact that it is computed as the average daily volume ``adv'' in shares, times the price ``p''. For simplicity we assume trading one percent of $\text{advp}_i$ will cause 20 basis points of market impact, with extension by linearity, meaning that 
\begin{equation} \label{eq:lambda_i} 
\lambda^i = 20 \times 10^{-4} \times \frac{1}{0.01 \text{advp}_i}.
\end{equation}
For very large trades (say, more than 0.05 $\text{advp}_i$), simple models such as \eqref{eq:lambda_i} break down. For this reason, we restrict our attention in this example to trades that are relatively small with respect to the anticipated volume. \medskip

One of the most challenging aspects of this study is simulating passive
execution, which we defined previously as a process of continually joining the queue on the near side of the limit order book until the order is filled, but never crossing the spread. \medskip

We are limited to the academic data sets available via the Wharton Research Data Services (WRDS). For this exercise, we used the New York Stock Exchange Trade and Quote (TAQ) database, which contains intraday transactions data on trades and quotes for all securities listed on the New York Stock Exchange (NYSE) and American Stock Exchange (AMEX), as well as Nasdaq National Market System (NMS) and SmallCap issues. \medskip

The TAQ database represents the aggregate inside quote for each exchange. Therefore, it includes both specialists and the public limit order book. Only having access to the consolidated feed, we construct a conservative simulation of when passive fills occur. Specifically, if we have a ``buy'' limit order (the entire process is similar for limit ``sell'' orders with ``bid'' replaced by ``ask'') which is simulated as existing in the queue on the bid side of the order book, when can we assume such an order was filled? Conservatively, if the order book changes and the new ask price is less or equal to the existing limit order price, we assume that markets would have cleared in the process of this change, and our limit order would have been filled, at least partially. We limit the amount of fill to the posted quantity at the new ask price. If this quantity is simulated to have been taken out, then no further fills are allowed to occur in the simulation until the price level changes. We assume that when the price level of the NBBO has changed, the liquidity is also replenished to the reported value at the new price level. This is a fairly conservative set of conventions; in reality, a larger number of passive fills could occur than merely the ones we simulate. This is because if there are multiple limit orders in the queue, one limit order can, of course, be filled without either bid or ask price levels changing. \medskip

Predicting the probability of a passive fill, denoted $f_i$ above, 
is equivalent to predicting the next transition of the limit order book and hence requires a model of limit order book dynamics. Indeed, such fill probabilities are one of the possible outputs of the very detailed model of \cite{cont2010stochastic} or the microstructure trading model presented in the previous section. As our data set is only the consolidated feed, we simply take $f_i = 0.1$ as the passive fill probability.

\subsection{Results}

As indicated above, we construct a market-neutral portfolio of $d = 2$ 
securities in which the long side is initially 1000 shares of IBM. 
Security $i = 1$ is IBM and $i = 2$ is AAPL. We estimate the security 
betas to the S\&P 500 (via regression on several years of daily data) as 
\begin{equation}\label{eq:betahat} 
\hat{\beta}_1 = 0.705, \quad 
\hat{\beta}_2 = 1.276 \, .
\end{equation}
We begin the simulation at 10:00 am on October 15, 2008, rather than 
immediately at the open since there are often outlier quotes, wide spreads 
and other effects around the open. The most recent midpoint price of IBM at 10:00 am was $p_1 = 93.06$ and for AAPL, $p_2 = 105.985$. \medskip

For convenience, we keep track of a cash balance for each position. The $n_1 = 1000$ shares of IBM are financed by borrowing $n_1 p_1 = $ USD 93,060 in cash and purchasing a position initially worth USD 93,060, so the net value 
(cash plus stock) of that position is initially zero. Similarly, the short position in AAPL is obtained by borrowing $n_2 = -485$ shares and immediately selling them for USD 51,403, and this position also initially has a net (cash + stock) value of zero. Note that with these holdings, \eqref{eq:betahat} implies that the portfolio's beta is 
\[ 
	n_1 p_1 \hat{\beta}_1 + n_2 p_2 \hat{\beta}_2 \approx 0 \, . 
\]
Any cash generated from further stock sales or cash used for further purchases of the same security is considered part of the separate cash balance allocated to that position. As prices change and as orders are filled, the values of each position will fluctuate. \medskip

Let $n_{i,t}$ denote the number of shares held in the $i$-th security at time $t$, and $p_{i,t}$ the latest midpoint price as of time $t$. 
Also, let $c_{i,t}$ denote the amount of cash (which can be positive or negative) attributed to the $i$-th security at time $t$, according to the accounting conventions outlined above. These variables change throughout the lifetime of the execution. \medskip

The value of a position is the number of shares held times the most 
recent midpoint price, plus the total amount of cash associated to the position, i.e. 
$n_{i,t} p_{i,t} + c_{i,t}$. 
The value of a portfolio is the sum of the values of all its positions, i.e. 
\begin{equation}\label{eq:value} 
\text{value}_t := \sum_{i=1}^d (n_{i,t} p_{i,t} + c_{i,t}).
\end{equation}
The value process \eqref{eq:value}, and especially its drift, is one measure of the execution's quality. If the value tends to drift downward, as in the AlwaysPassive model detailed below, then the execution desk is losing money due to slippage. This is perhaps the typical situation -- one expects execution to have associated costs. A particularly pleasant situation arises when the drift of the portfolio value process \eqref{eq:value} is zero, as in Figure \ref{fig:BenvenisteRitterPort}, and it is possible that with very good microstructure alphas added to the generalized momenta, the drift could even become positive. All monetary values are reported in USD. The predicted daily volumes are estimated to 
\[ 
\text{advp}_1 \approx 1.16 \times 10^9, \quad 
\text{advp}_2 \approx 6.13 \times 10^9 \, .
\]
The covariance matrix is 
\[ 
\Sigma = 10^{-4} \times \left(
\begin{array}{cc}
15.5728 & 17.7558 \\
17.7558 & 28.6519 \\
\end{array}
\right),
\]
which implies a correlation of 0.84 among the two assets and daily volatilities of approximately 3.9\% and 5.4\%.  \medskip

The output of our algorithm is the instantaneous 
aggression level: aggressive, passive, or wait. 
There is not a unique benchmark to gauge such an algorithm's 
performance, but it is sensible to compare a complicated method 
of choosing the aggression level to a simple method for choosing 
the aggression level, to see if the additional complexity is justified. 
Hence one could compare it to a constant aggression level -- always passive. \medskip 

Figure \ref{fig:AlwaysPassivePort}  reveals that, as the market was falling, the passive ``buy'' orders in AAPL were all filled very quickly, while 
unsurprisingly the ``sell'' orders in IBM were filled very slowly, and indeed were not even finished by the end of the trading day. This drove the 
Gross Market Value (gmv) down while pushing the net and beta higher, where we define 
\begin{equation}\label{eq:beta-t}
\beta_t := \sum_i n_{i,t} p_{i,t} \hat{\beta}_i, 
\ \ 
\text{net}_t := \sum_i n_{i,t} p_{i,t}, 
\ \ 
\text{gmv}_t := \sum_i |n_{i,t} p_{i,t}|, 
\end{equation}
with $\hat{\beta}_i$ given by \eqref{eq:betahat}. Thus the portfolio had $\beta_t > 0$ in a falling market. Note that the losses incurred in this manner do not become gains if the sign of the market move is reversed; they remain losses irrespective of the market's direction. In a rising market, the ``always passive'' model would have the same problem: the ``sell'' orders would be filled quickly, the ``buy'' orders would linger, and the portfolio would build up negative beta in a rising market. 
\vskip 0.05in
\begin{figure}[H]
	\includegraphics[width=0.45\textwidth,height=0.25\textwidth]{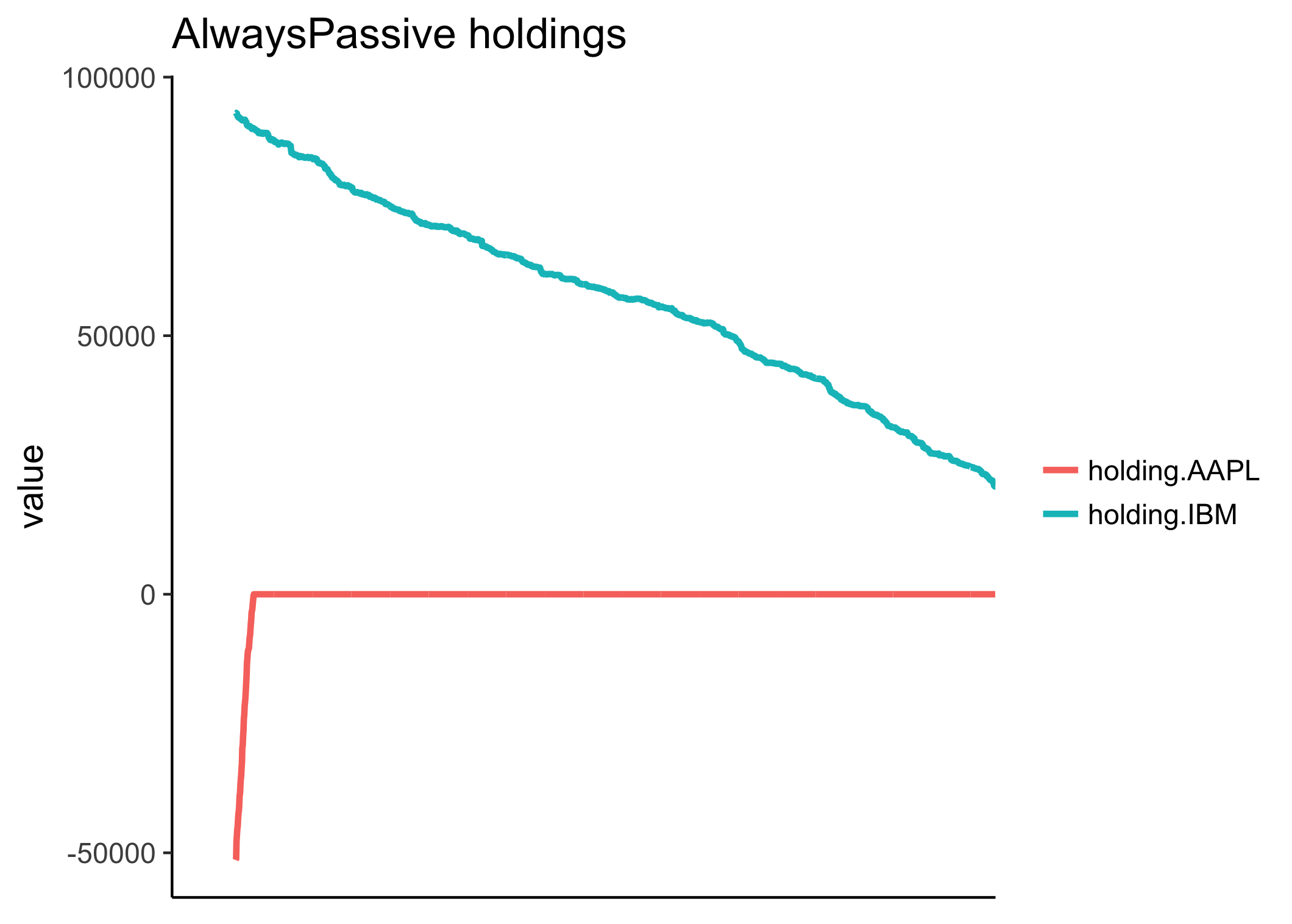}\qquad\includegraphics[width=0.45\textwidth,height=0.25\textwidth]{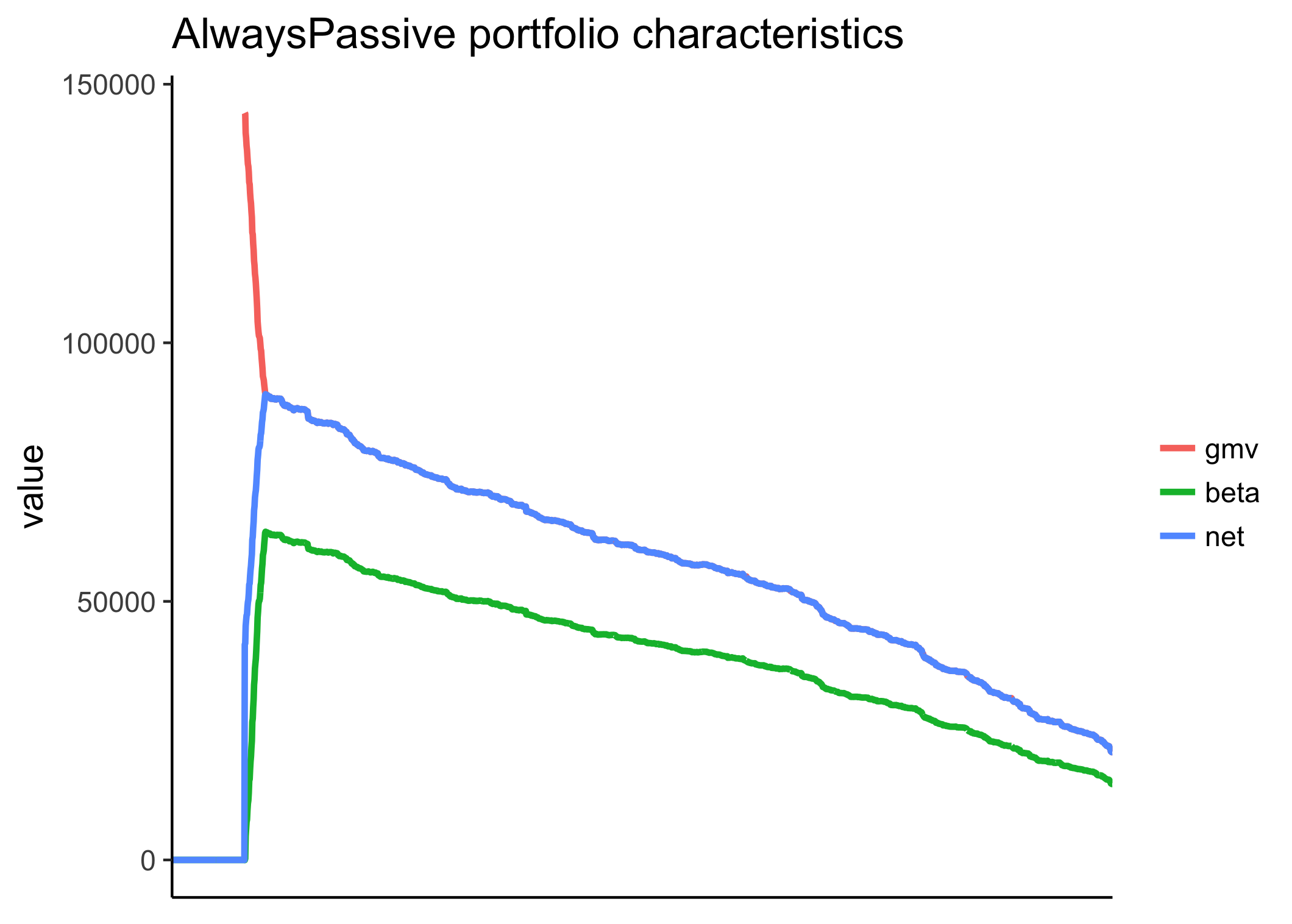}
	\caption{Portfolio holdings in the ``always passive'' model, and 
		portfolio characteristics: gross market value (gmv), net, 
		and $\beta_t$ given by \eqref{eq:beta-t}. }
	\label{fig:AlwaysPassivePort} 
\end{figure}
We now show the analogous graphs for the simplest version of our execution model developed in the previous section. Note that the model retains a fairly small beta exposure throughout the lifetime of the execution. This is because CAPM beta is also a factor in the APT risk model, and the generalized momenta point along the gradient of the Hamilton-Jacobi-Bellman value function and hence drive trading towards the optimal value of 
multiperiod utility (including the risk term). This is the key advantage of our model over simpler execution algorithms. 
\vskip 0.05in
\begin{figure}[H]
	\includegraphics[width=0.45\textwidth,height=0.25\textwidth]{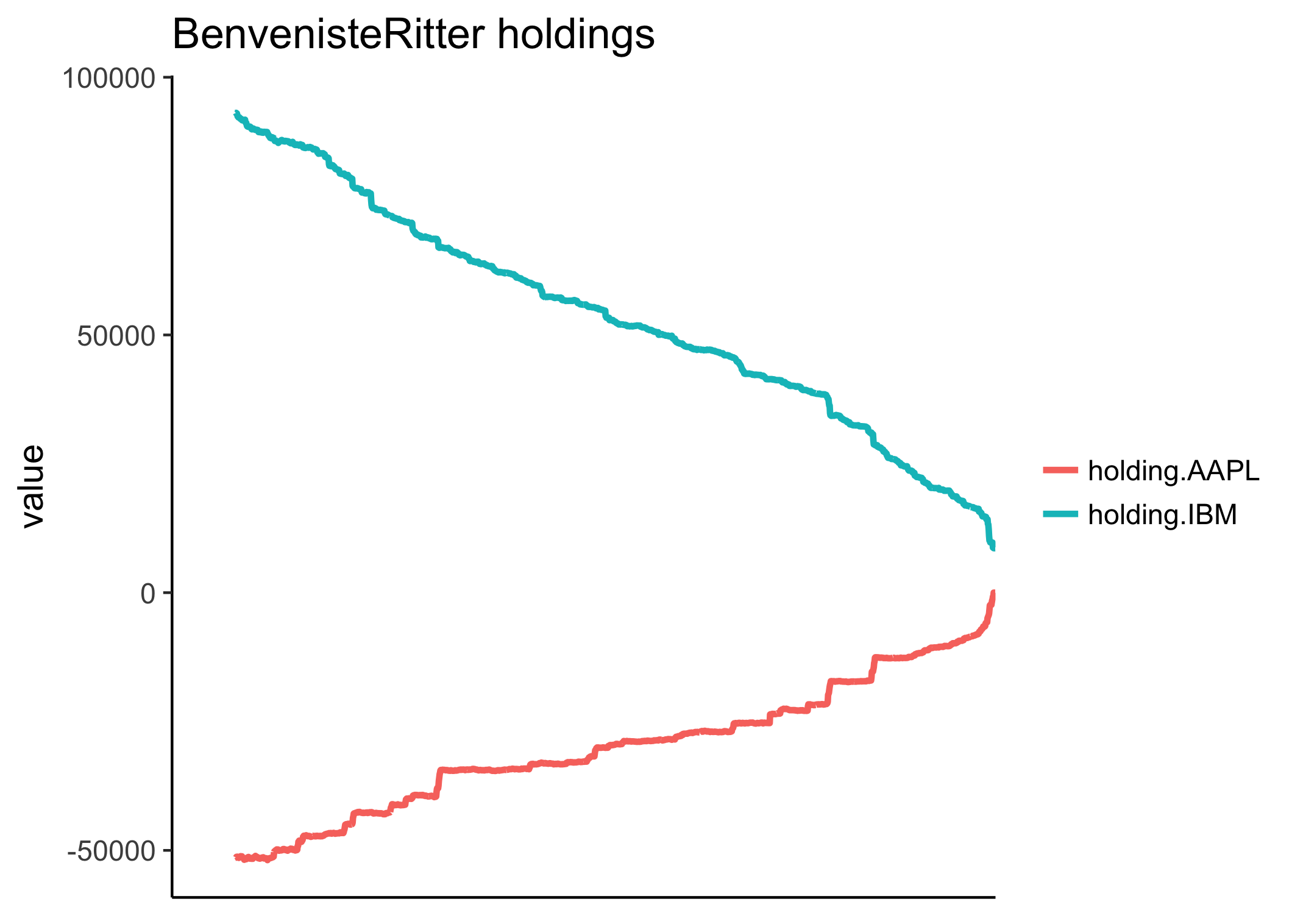}\qquad\includegraphics[width=0.45\textwidth,height=0.25\textwidth]{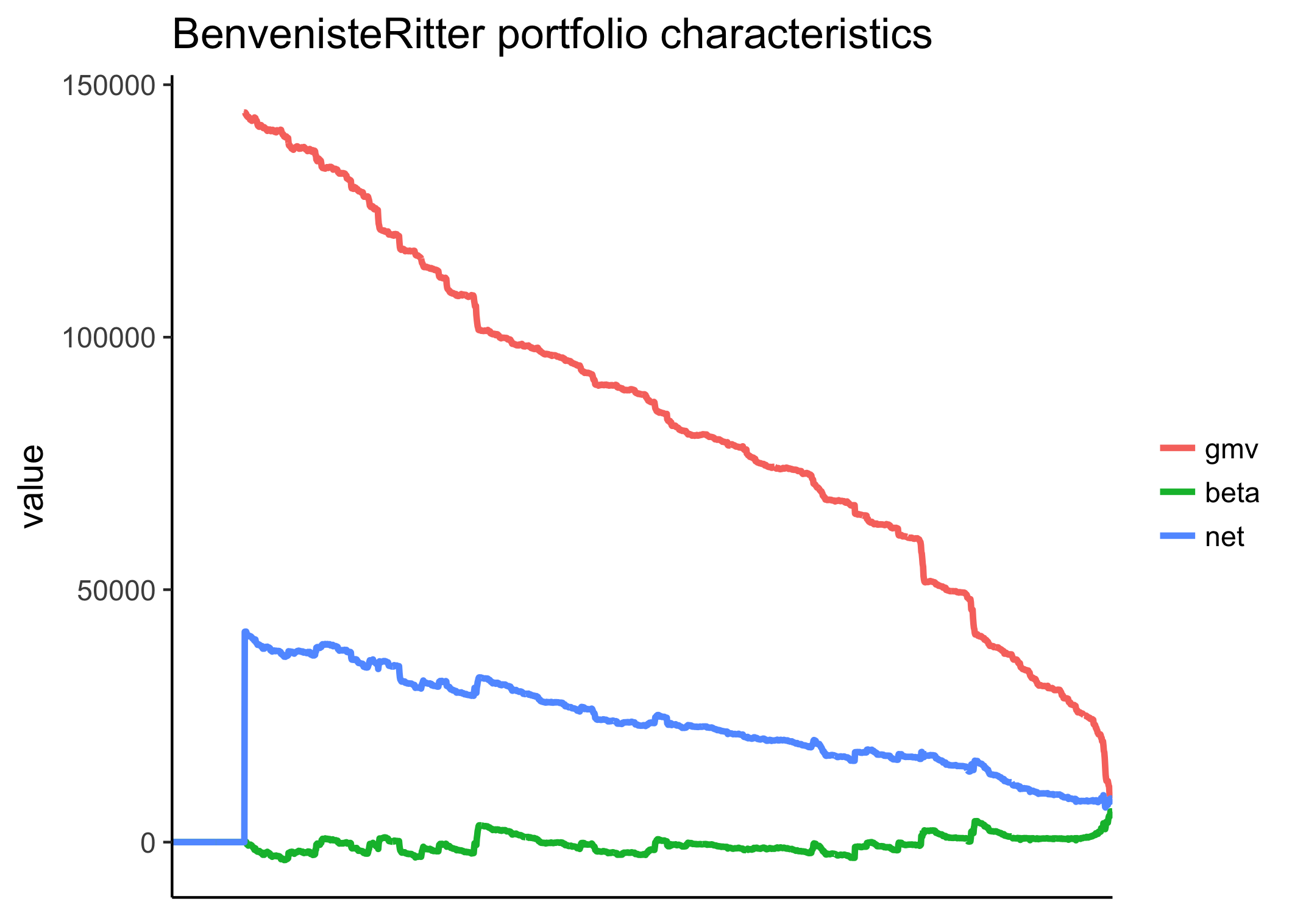}
	\caption{Portfolio holdings in the sophisticated model, and 
		portfolio characteristics: gross market value ($\mathrm{gmv}_t$), 
		$\mathrm{net}_t$, and $\beta_t$ given by \eqref{eq:beta-t}. }
	\label{fig:BenvenisteRitterPort} 
\end{figure}
Finally, we consider the portfolio value over the lifetime of the execution. Note that in our model, the value process \eqref{eq:value} is approximately driftless, which as explained above is a desirable property, and outperforms the ``always passive'' value process realization. In particular, in our model $\text{value}_t$ is able to avoid negative drift in a falling
market precisely because the portfolio remains approximately beta-neutral. In a portfolio with many assets (large $n$), our method would allow it to remain approximately neutral to all factors in the APT model. 
\vskip 0.05in
\begin{figure}[H]
	\centering\includegraphics[width=0.8\textwidth,height=0.35\textwidth]{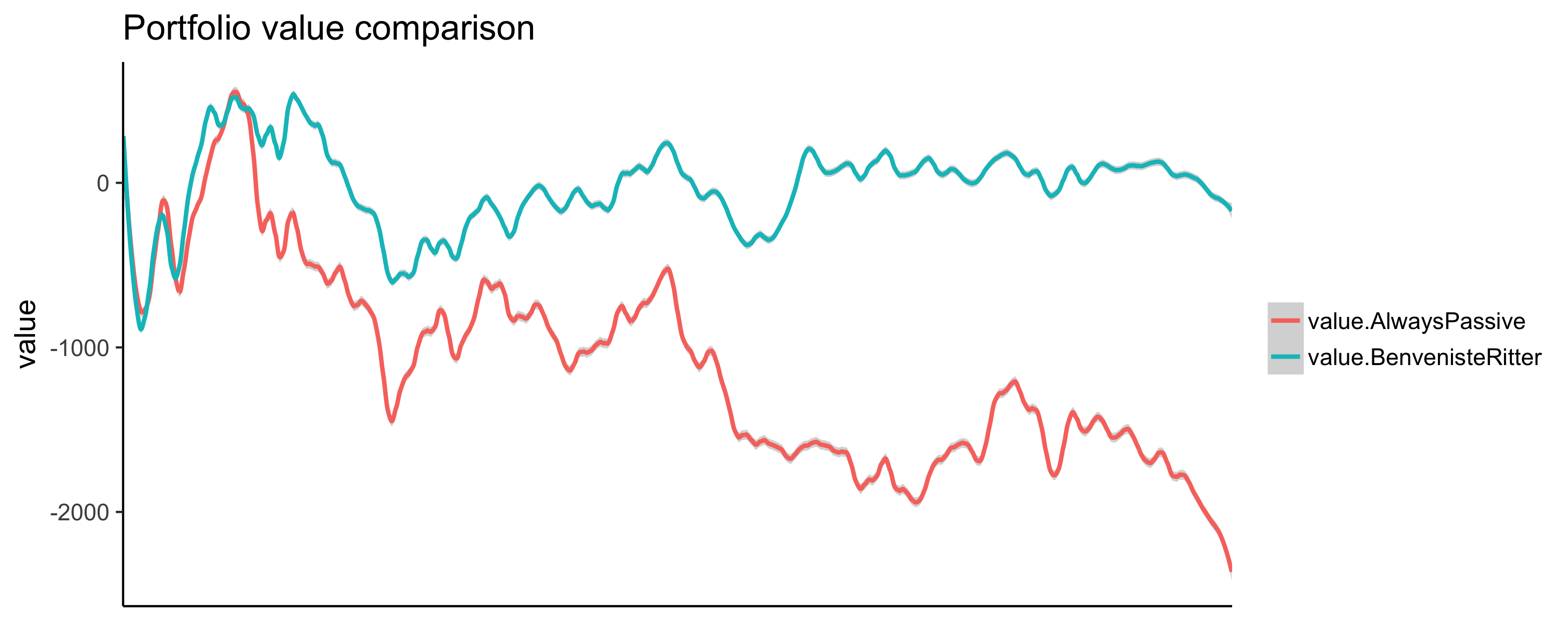}
	\caption{Portfolio value \eqref{eq:value} over the lifetime of the execution, for
		both execution methods.}
	\label{fig:ValueComparison} 
\end{figure}
\vskip 0.05in
The difference in Figure \ref{fig:ValueComparison}  is both statistically and economically significant. The t-statistic for the difference is about 78, hence significant at the 99.999\% level. Moreover, the dollar value of the difference between the two methods is about 1.5\% of the initial gross market value to be liquidated.

\section{Conclusion}

In this paper, we present a framework to perform optimal trading, taking into account market microstructure and a long-term trading schedule without the use of optimal control. This approach relies on the use of the 
generalized momenta $p = \nabla V(t,q)$ as the effective microstructure alpha. We show that a myopic agent sending only market orders with such alpha will minimize the error with respect to the long-term trading schedule. Moreover, when we add the possibility of passive execution, the long-term alpha can be chosen as a transformation of the generalized momenta $p$. We also present a general microstructure trading framework for the multi-asset multi-venue optimal trading problem. For a parsimonious model of fill probabilities, the effective microstructure alpha can be computed in closed form. We apply the same heuristics to derive an optimal market-making model that is tractable for a large number of assets and venues.  \medskip

Based on the dual formulation of the classic Almgren-Chriss optimization problem, this simple heuristics has wide-ranging practical implications. In addition to bridging the gap between order placement and scheduling, it simplifies optimal trading problems that are usually intractable using optimal control due to the high-dimensional Hamilton-Jacobi-Bellman equation resulting from the control problem. This is of particular importance for a quantitative execution desk wishing to trade a high number of cross-listed assets. It opens up many avenues for future exploration. One set of projects is to consider trading problems beyond the typical buy-side utility-maximization, which can still be viewed within the unifying framework of a myopic risk-neutral wealth-maximizer, whose microstructure alphas are aligned with the value function gradient.

\bibliography{ejb-gr.bib}

\begin{thebibliography}{17}
\providecommand{\natexlab}[1]{#1}
\providecommand{\url}[1]{\texttt{#1}}
\expandafter\ifx\csname urlstyle\endcsname\relax
  \providecommand{\doi}[1]{doi: #1}\else
  \providecommand{\doi}{doi: \begingroup \urlstyle{rm}\Url}\fi

\bibitem[Almgren and Chriss(2001)]{almgren2001optimal}
R.~Almgren and N.~Chriss.
\newblock Optimal execution of portfolio transactions.
\newblock \emph{Journal of Risk}, 3:\penalty0 5--40, 2001.

\bibitem[Arrow(1963)]{arrow1963}
K.~J. Arrow.
\newblock Liquidity preference, lecture vi in ``lecture notes for economics
  285, the economics of uncertainty'', pp 33-53.
\newblock 1963.

\bibitem[Avellaneda and Stoikov(2008)]{avellaneda2008high}
M.~Avellaneda and S.~Stoikov.
\newblock High-frequency trading in a limit order book.
\newblock \emph{Quantitative Finance}, 8\penalty0 (3):\penalty0 217--224, 2008.

\bibitem[Baldacci and Manziuk(2020)]{baldacci2020adaptive}
B.~Baldacci and I.~Manziuk.
\newblock Adaptive trading strategies across liquidity pools.
\newblock \emph{arXiv preprint arXiv:2008.07807}, 2020.

\bibitem[Bergault et~al.(2020)Bergault, Evangelista, Guéant, and
  Vieira]{bergault2020closedform}
P.~Bergault, D.~Evangelista, O.~Guéant, and D.~Vieira.
\newblock Closed-form approximations in multi-asset market making, 2020.

\bibitem[Cartea et~al.(2014)Cartea, Jaimungal, and Ricci]{cartea2014buy}
{\'A}.~Cartea, S.~Jaimungal, and J.~Ricci.
\newblock Buy low, sell high: A high frequency trading perspective.
\newblock \emph{SIAM Journal on Financial Mathematics}, 5\penalty0
  (1):\penalty0 415--444, 2014.

\bibitem[Cartea et~al.(2015)Cartea, Jaimungal, and
  Penalva]{cartea2015algorithmic}
{\'A}.~Cartea, S.~Jaimungal, and J.~Penalva.
\newblock \emph{Algorithmic and high-frequency trading}.
\newblock Cambridge University Press, 2015.

\bibitem[Cont and Kukanov(2017)]{cont2017optimal}
R.~Cont and A.~Kukanov.
\newblock Optimal order placement in limit order markets.
\newblock \emph{Quantitative Finance}, 17\penalty0 (1):\penalty0 21--39, 2017.

\bibitem[Cont et~al.(2010)Cont, Stoikov, and Talreja]{cont2010stochastic}
R.~Cont, S.~Stoikov, and R.~Talreja.
\newblock A stochastic model for order book dynamics.
\newblock \emph{Operations Research}, 58\penalty0 (3):\penalty0 549--563, 2010.

\bibitem[Dreyfus(1960)]{dreyfus1960dynamic}
S.~E. Dreyfus.
\newblock Dynamic programming and the calculus of variations.
\newblock \emph{Journal of Mathematical Analysis and Applications}, 1\penalty0
  (2):\penalty0 228--239, 1960.

\bibitem[Glosten and Milgrom(1985)]{glosten1985bid}
L.~R. Glosten and P.~R. Milgrom.
\newblock Bid, ask and transaction prices in a specialist market with
  heterogeneously informed traders.
\newblock \emph{Journal of financial economics}, 14\penalty0 (1):\penalty0
  71--100, 1985.

\bibitem[Gu{\'e}ant(2016)]{gueant2016financial}
O.~Gu{\'e}ant.
\newblock \emph{The Financial Mathematics of Market Liquidity: From optimal
  execution to market making}, volume~33.
\newblock CRC Press, 2016.

\bibitem[Gu{\'e}ant and Manziuk(2019)]{gueant2019deep}
O.~Gu{\'e}ant and I.~Manziuk.
\newblock Deep reinforcement learning for market making in corporate bonds:
  beating the curse of dimensionality.
\newblock \emph{Applied Mathematical Finance}, 26\penalty0 (5):\penalty0
  387--452, 2019.

\bibitem[Gu{\'e}ant et~al.(2013)Gu{\'e}ant, Lehalle, and
  Fernandez-Tapia]{gueant2013dealing}
O.~Gu{\'e}ant, C.-A. Lehalle, and J.~Fernandez-Tapia.
\newblock Dealing with the inventory risk: a solution to the market making
  problem.
\newblock \emph{Mathematics and financial economics}, 7\penalty0 (4):\penalty0
  477--507, 2013.

\bibitem[Guilbaud and Pham(2013)]{guilbaud2013optimal}
F.~Guilbaud and H.~Pham.
\newblock Optimal high-frequency trading with limit and market orders.
\newblock \emph{Quantitative Finance}, 13\penalty0 (1):\penalty0 79--94, 2013.

\bibitem[Ho and Stoll(1981)]{ho1981optimal}
T.~Ho and H.~R. Stoll.
\newblock Optimal dealer pricing under transactions and return uncertainty.
\newblock \emph{Journal of Financial economics}, 9\penalty0 (1):\penalty0
  47--73, 1981.

\bibitem[Pratt(1964)]{pratt1964risk}
J.~W. Pratt.
\newblock Risk aversion in the small and in the large.
\newblock \emph{Econometrica: Journal of the Econometric Society}, pages
  122--136, 1964.

\end{thebibliography}

\end{document}